\documentclass[article,12pt,dvipsnames]{article}
\usepackage{amsmath,amssymb,cite}
\usepackage{amsthm,xcolor}
\usepackage{graphics,epsfig,color}
\usepackage[mathscr]{euscript}
\usepackage{thmtools}
\usepackage{dsfont}
\usepackage{xfrac}

\usepackage{comment}
\usepackage{enumerate}

\usepackage{graphicx}
\usepackage{graphics}
\usepackage{caption}
\usepackage{subcaption}
\usepackage{hyperref}
\usepackage{tikz}
\usetikzlibrary{arrows}

\usepackage{tabularx}
\usepackage{bm}
\usepackage{multirow}
\usepackage{pifont}

\usepackage{ bbold }

\theoremstyle{plain}
\newtheorem{theorem}{Theorem}[section]
\newtheorem{proposition}[theorem]{Proposition}
\newtheorem{lemma}[theorem]{Lemma}
\newtheorem{corollary}[theorem]{Corollary}

\theoremstyle{definition}
\newtheorem{definition}[theorem]{Definition}

\theoremstyle{remark}
\newtheorem{remark}[theorem]{Remark}
\newtheorem{example}[theorem]{Example}

\newtheorem{axiom}{Axiom}
\newtheoremstyle{colored}
  {3pt}{3pt}%
  {\color{blue}}
  {}
  {\color{red}\bfseries}
  {.}
  { }
  {}
\theoremstyle{colored}

\numberwithin{equation}{section}
\numberwithin{theorem}{section}

\renewcommand{\epsilon}{\varepsilon}

\renewcommand{\hat}{\widehat}

\newcommand\rmv[1]{\textcolor{BrickRed}{#1}}
\newcommand\mic[1]{\textcolor{Orange}{#1}}

\title{Coopetitive Index: a measure of cooperation and competition in coalition formation} 
\author{Michele Aleandri, Marco Dall'Aglio}
\date{November 19, 2025}

\begin{document}

\maketitle

\begin{abstract}
    We extend the coopetition index introduced by Aleandri and Dall’Aglio (2025) for simple games to the broader class of monotone transferable-utility (TU) games and to all non-empty coalitions, including singletons. The new formulation allows us to define an absolute coopetition index with a universal range in \([-1,1]\), facilitating meaningful comparisons across coalitions.
    
We study several notable instances of the index, including the Banzhaf, Uniform Shapley, and Shapley–Owen coopetition indices, and we derive explicit formulas that connect coopetition to classical semivalues. Finally, we provide axiomatic characterizations of the Uniform Shapley and Shapley–Owen versions, showing that each is uniquely determined by linearity, symmetry over pure bargaining games, external null player neutrality, and a contraction axiom reflecting its internal distribution. These results position the coopetition index as a versatile tool for quantifying the cooperative and competitive tendencies of coalitions in TU-games.
\end{abstract}
\section{Introduction}
Cooperation and competition coexist in many economic, political, and social environments where groups of agents form coalitions to pursue common objectives. Understanding the internal cohesion of a coalition, and whether its members behave cooperatively or competitively, is therefore essential for interpreting outcomes in cooperative game theory.
In previous work, Aleandri and Dall’Aglio \cite{aleandri2025coopindex} introduced a coopetition index for simple games, capturing the tendency of a coalition to act as a unified entity or to split into subcoalitions. The index relies on two probability families: an internal distribution describing how the coalition may divide itself, and an external distribution describing how the coalition interacts with the rest of the players.

In this paper we extend the coopetition index in two directions. First, we generalize its definition to the broader class of monotone TU-games, allowing us to evaluate cooperation and antagonism in settings where coalition worths are not restricted to binary values. Second, we expand the domain to include all non-empty coalitions, in particular singletons, which naturally bridge the coopetition index with generalized semivalues.
We study the structural properties of the extended index, establish tight bounds, and show how it relates to classical group values. Finally, we provide axiomatic characterizations of two important instances: the Uniform Shapley and the Shapley–Owen coopetition indices.

\section{Basic Notation}\label{sec:notation}

We denote the cardinality of a finite set with the corresponding small letter: e.g. \(n=|N|\), \(s=|S| \), \(t=|T|\). Given a set $A\subseteq N$, its complement is defined as $A^c=N\setminus A$. Also, 
we will often omit braces for singletons, e.g. writing $v(i), S \backslash i$ instead of $v(\{i\}), S \backslash\{i\}$. Similarly, we will write $i j, ijk$ instead of $\{i, j\}, \{i,j,k\}$. 

\begin{definition} A cooperative game with transferable utility or TU-game is a pair $(N, v)$, where $N=\{1,2, \ldots, n\}$ with $n \in \mathbb{N}$ is the set of players, and $v$ is a function assigning to each coalition $S$, i.e. to each subset $S \subseteq N$, a real number $v(S)$, such that $v(\emptyset)=0$. The function $v$ is called the characteristic function and $v(S)$ is called the worth of $S$. The coalition $N$ is called the grand coalition. The TU-game \(v\) is {\em monotone} if \(v(S) \leq v(T)\) whenever \(S \subseteq T \subseteq N\). A monotone TU-game is {\em simple} if \(v(S) \in \{0,1\}\), for \(S \subseteq N\), and \(v(N)=1\).
\end{definition}

 The marginal contribution of coalition $S$ to coalition $T$, $T\subseteq N\setminus  S$, is defined as:
	\[
	v_S'(T)= v(S \cup T )-v(T ). 
	\]
	 A player $i$ such that $v_i'(T)=0$, for all $T\subseteq N\setminus\{i\}$, is called \emph{null player}.\\
    
 Call \(\mathcal{G}^N\) the class of games defined on a set of players \(N\). Given a game \(v \in \mathcal{G}^N\) and a coalition \(A \subseteq N\), the game \(v^{-A} \in \mathcal{G}^{N \setminus A}\) (\(v\) without \(A\) ) is the restriction of \(v\) to \(N \setminus A\):
\[
v^{-A}(S)=v(S), \qquad \mbox{for any } S \subseteq N \setminus A .
\]
For any coalition \( B \subseteq A\), the game \(v^{-A}_{\cup B} \in \mathcal{G}^{N \setminus A}\) ( $v$ without \(A\) in the presence of \(B\) ) is defined as
\[
v^{-A}_{\cup B}(S)=v(S \cup B)-v(B),    \qquad \mbox{for any } S \subseteq N \setminus A .
\]   
	We recall the definitions of the probabilistic generalized value  and generalize semivalue given by Marichal et al. \cite{marichal2007axiomatic}. 

\begin{definition}
  Given a finite set $N$, consider a family of probability distributions \(q^N=\{q_S^N\}_{S\subseteq N}\), where, for each $S\subseteq N$, $q_S^N$ is a probability distribution over the set $N\setminus S$. A \emph{probabilistic generalized value} of a coalition $S \subseteq N$ is defined by
$$
\Phi_{q^N}^v(S)=\sum_{T \subseteq N \backslash S} q^N_S(T)v_S'(T).
$$
\\
A \emph{generalized semivalue} is a probabilistic generalized value such that, additionally, for any $S \subseteq N$, the probability $q_S^N(T)$, $\forall T\subseteq N\setminus S$, depends only on the cardinalities of the coalitions $S, T$, and $N$, i.e., for any $s \in\{0, \ldots, n\}$, there exists a family of nonnegative real numbers $\left\{q_s^n(t)\right\}_{t=0, \ldots, n-s}$ fulfilling
$$
\sum_{t=0}^{n-s}\binom{n-s}{t} q_s^n(t)=1,
$$
such that, for any $S \subseteq N$ and any $T \subseteq N \backslash S$, we have $q_S^N(T)=q_s^n(t)$.
\end{definition}

Notable examples of generalized semivalues are the Shapley Generalized Value, obtained by setting \(q_s^n(t)=\frac{(n-s-t)!t!}{(n-s-1)!}\) and the Banzhaf Generalized value when \(q_s^n(t)=2^{s-n}\).

 Here and throughout, when no ambiguity arises, we omit the dependence on $N$, and \(q\) will denote a family of probability distributions $\{q_{S}\}_{S\subseteq N}$, where each $q_S$ is a probability distributions over $2^{N\setminus S}$.

\section{Extension of the coopetition index to singletons}
The coopetition index of a coalition was introduced in \cite{aleandri2025coopindex} for coalitions of 2 or more players in the context of simple games. In this section we extend the formula to singletons and to monotone TU-games. The first extension is achieved by means of a different formulation of the  formula already available.

Take a coalition $S\subseteq N$, with $|S|\geq 2$,  and call $\Pi_2(S)$ the set of non-trivial (i.e.\ non-empty) $2$-partitions, i.e.  
$$\Pi_2(S)=\big\{ \pi=\{S_1,S_2\}:\,\emptyset \neq S_1,S_2\subset S,\, S_1\cup S_2 = S \mbox{ and } S_1\cap S_2 =\emptyset  \big\}.$$  We begin by reconsidering the attitude of a coalition towards another non-overlapping one. 

\begin{definition}
    \label{def:new_Attitude}
Given a finite set $N$, take a family of probability distribution  $p^N=\{p_S^N\}_{S\subseteq N}$, where each $p_S^N$ is a probability distributions over the set $\Pi_2(S)$. The \emph{attitude} of coalition $S \subseteq N$ towards $T \subseteq N \setminus S$ under the family \(p^N\) is defined as
\begin{multline}
    \label{eq:new_Attitude}
\mathcal{A}^v_{p^N(S,T)} = v(S \cup T) - v(T) \\ -\sum_{\pi=\{S_1,S_2\} \in \Pi_2(S)} p_S^N(\pi)\big( v(S_1 \cup T) + v(S_2 \cup T) - 2v(T) \big),
\end{multline}
with the implicit assumption that when \(|S|=1\), \(\Pi_2(S)=\emptyset\) and the sum is null. In such a case, \(\mathcal{A}^v_{p^N}(S,T) = v(S \cup T) - v(T)\) whatever \(p\) we are considering.
\end{definition}

In what follows, when no ambiguity arises, we omit the dependence on $N$, and \(p\) will always indicate a family of probability distributions \(\{p_S\}_{S\subseteq N}\), where each \(p_S\) is a probability distribution on \(\Pi_2(S)\). 
The definition can be written more explicitly by singling out the case of singletons from all the other ones.

\begin{proposition}
    \label{prop:attitude_SingOut}
    For a given family of distributions \(p\), $S \subseteq N$ and $T \subseteq N \setminus S$, we can write
    \begin{equation}
        \label{eq:attitude_SingOut}
        \mathcal{A}^v_p(S,T) = \begin{cases}
            v(S \cup T) - v(T), & \mbox{if } |S|=1;
            \\
            v(S \cup T) + v(T) - \sum_{\pi=\{S_1,S_2\} \in \Pi_2(S)} p_S(\pi)\left( v(S_1 \cup T) + v(S_2 \cup T)   \right), & \mbox{otherwise.} 
        \end{cases}
    \end{equation}
\end{proposition}
In its first appearance, the attitude was defined for simple games. The notion can be applied to monotone TU-games at no additional cost. Some care is needed in its interpretation, starting from the definition of two important cases for the index.

\begin{definition}
    Consider a monotone TU-game \(v\) defined on \(N\). If two coalitions \(S\) and \(T \subseteq N \setminus S\) are such that \(v(S \cup T) - v(T) > 0\), coalition \(S\) is called a \(T\)\emph{-contributing coalition}. If the \(T\)-contributing coalition \(S\) is such that \(v(P \cup T) = v(T)\) for every proper subcoalition \(P \subsetneq S\), \(S\) is an {\em essential} \(T\)-contributing coalition. If, instead, \(v(P \cup T) = v( S \cup T)\) for every nonempty subcoalition \(\emptyset \neq P \subseteq S\), then \(S\) is a {\em fully complementary} \(T\)-contributing coalition.
\end{definition} 

The following result shows that the previous definition singles out the range endpoints for the attitude.
\begin{proposition} \label{prop:-1+1}
    Let \(S \subseteq N\)  and $T\subseteq N \setminus S$, with \(S\) a \(T\)-contributing coalition. Let \(p_S\) be a probability distribution with full support over $\Pi_2(S)$, i.e. \(p_S(\pi)>0\) for every \(\pi \in \Pi_2(S)\). Then 
    \begin{itemize}
        \item[i)] \(S\) is an essential \(T\)-contributing coalition if and only if $ \mathcal{A}^v_{p}(S,T)=v(S \cup T) - v(T)$
        \item[ii)] \(S\) is a fully complementary \(T\)-contributing coalition if and only if $ \mathcal{A}^v_{p}(S,T)=v(T) - v(S \cup T)$. 
    \end{itemize}
\end{proposition}
\begin{proof}
If \(|S|=1\), then \(S\) is always an essential \(T\)-contributing coalition. Assume now \(|S| \geq 2\). Now if $\mathcal{A}^v_{p}(S,T)=v(S \cup T) - v(T)$, it means that \(v(S_1 \cup T) = v(S_2 \cup T) = v(T) \) for all $\{S_1,S_2\}\in\Pi_2(S)$. This implies that $S$ is an essential \(T\)-contributing coalition.  The opposite case, i.e.\  $\mathcal{A}^v_{p}(S,T)=v(T)-v(S \cup T)$ is treated similarly.  The "if" part is easy to obtain and left to the reader.
\end{proof}

We are now ready to define the indices that measure the cooperation and antagonism of players to form a coalition. 
\begin{definition}
    \label{def:gCoopIndex}
Take two families of  probability distribution $p$ and $q$. The \emph{coopetition index } for coalition $S\subseteq N$ is defined as
\begin{equation*}
    \mathcal{C}^v_{p,q}(S) = \sum_{T\subseteq N\setminus S} q_S(T)\mathcal{A}^v_{p}(S,T).
\end{equation*}
\end{definition}

To distinguish between the two families of distributions that define the index, we will refer to \(p\) as the {\em internal} distribution, and to \(q\) as the {\em external} probability distribution family.

The new definition of coopetition index includes the singletons. For these sets the family \(p\) has empty support, and the coopetition index coincides with the semivalues originated from the same family \(q\). If \(S\) has two or more players, the attitude of \(S\) towards another coalition \(T\) ranges between \(-v'_S(T) \) and \(v'_S(T)\). Not surprisingly, coalitions of size 2 or larger are bound by the generalized group value defined by the same family \(q\) and its opposite. We gather these findings in the next result.

\begin{proposition} \label{prop:coptbounds}
    For any game \(v\), any families \(p\) and \(q\), and any \(i \in N\), the following holds: 
    \begin{enumerate}[i)]
        \item If \(S=i\), then
            \begin{equation}
        \label{eq:singleton_coincidence}
\mathcal{C}^v_{p,q}(i)=\Phi^v_q(i).        
    \end{equation}
    \item  If $|S|\geq 2$, then
    \begin{equation*}
        -\Phi^v_q(S) \leq \mathcal{C}^v_{p,q}(S) \leq \Phi^v_q(S) \; .
    \end{equation*}
    \end{enumerate}
\end{proposition}
\begin{proof}
For \(S=i\), simply apply the definition. For \(|S| \geq 2 \)
consider the same reasoning of the proof of Proposition 3.5 in Aleandri and Dall'Aglio \cite{aleandri2025coopindex}.
\end{proof}

The bounds are attainable in special circumstances.
\begin{corollary}
\label{cor:attain}
For any given pair of families of probability distributions \(q\) and \(p\), the following holds:
    \begin{enumerate}[i)]
        \item If every time \(S\) is \(T\)-contributing for some coalition \(T \subseteq N \setminus S\) it is essential \(T\)-contributing for the same coalition, then \(\mathcal{C}_{pq}^v(S) = \Phi_q^v(S)\).
        \item If every time \(S\) is \(T\)-contributing for some coalition \(T \subseteq N \setminus S\) it is fully complementary \(T\)-contributing for the same coalition, then \(\mathcal{C}_{pq}^v(S) = - \Phi_q^v(S)\).
    \end{enumerate}
\end{corollary}

Measuring both the generalized group value and the coopetition index enables us to give an informative picture of the coalition strength and cohesion. Additionally, the previous result also shows that a ratio between the two indices provides an absolute coopetition index which has a common range for any coalition.

\begin{definition}
Take any pair of families of probability distributions \(q\) and \(p\). The \emph{absolute coopetition index } for coalition $S\subseteq N$ is defined as
\begin{equation*}
    \hat{\mathcal{C}}^v_{p,q}(S) = \begin{cases}
        \cfrac{\mathcal{C}^v_{p,q}(S)}{\Phi^v_q(S)}& \mbox{ if } \Phi^v_q(S)\neq 0\\
        0 & \mbox{otherwise}
    \end{cases}
\end{equation*}
where \(\Phi^v_q(S)\) is the probabilistic generalized value of \(S\).
\end{definition}

Clearly, Proposition \ref{prop:coptbounds} implies that, for every \(i \in N\), \(\hat{\mathcal{C}}^v_{p,q}(i)=1\), while for every \(S \subseteq N\) with \(|S| \geq 2\),
\[
-1 \leq \hat{\mathcal{C}}^v_{p,q}(S) \leq 1,
\]
and Corollary \ref{cor:attain} provides the conditions under which the bounds are attained.\\

In summary, the coopetition index offers a quantitative measure of the internal dynamics within a coalition, revealing whether the members are more inclined towards unified action (cooperation) or internal division and individual pursuits (competition), with the absolute index providing a standardized scale from -1 to 1. Its interpretation is enriched by considering the coalition's power, often measured by the generalized group value.

\subsection{Most Significant cases}
In Aleandri and Dall'Aglio \cite{aleandri2025coopindex} some coopetition indices have already been defined and used -- both in theory and in applications. The given definitions rely on a family of internal distributions \(p\) and external ones \(q\) that are widely employed in the context of cooperative games. To start with, we consider distribution families that are uniform over their support. The uniform distribution family \(q^u \) defines, for each \(S \subseteq N\), a distribution over the power set of \(N\setminus S\) that  is simply given by
\begin{equation}
\label{def:q_unif}
    q^u_S(T)=\frac{1}{2^{n-s}}, \qquad \mbox{for } T \subseteq N \setminus S .
\end{equation}
    The uniform family \(p^u\) defines, for each \(S \subseteq N\), a distribution over \(\Pi_2(S)\) given by
    \begin{equation}
        \label{def:p_unif}
 p^u_S(T)=\frac{1}{2^{s-1}-1}, \qquad \mbox{for } \pi \in \Pi_2(S) .
    \end{equation}
The combination of the two families yields the Banzhaf coopetition index
\(\mathcal{C}_{Bz}^v = \mathcal{C}_{p^u,q^u}^v  \).

An alternative proposal relies on permutation-induced families. The external permutation-induced family \(q^r\) defines, for each \(S \subseteq N\), a distribution over the power set of \(N\setminus S\) given by
\begin{equation}
\label{def:q_perm}
    q^r_S(T)=\frac{t!(n-s-t)!}{(n-s+1)!}, \qquad \mbox{for } T \subseteq N \setminus S .
\end{equation}
 The permutation-based family \(p^r\) defines, for each \(S \subseteq N\), a distribution over \(\Pi_2(S)\) given by
    \begin{equation}
        \label{def:p_perm}
 p^r_S(T)=\frac{2 r!(s-r)! }{(s-1) s!}, \qquad \mbox{for any } \pi=\{R,S\setminus R\} \in \Pi_2(S).
    \end{equation}
The combination of both permutation-induced families delivers the Shapley-Owen coopetition index defined as \(\mathcal{C}^v_{SO}=\mathcal{C}^v_{p^r,q^r}\).

In what follows, we will consider another index which takes one family from each group: the external permutation based family \(q^r\) will be combined with the internal uniform distribution family $p^u$.

\begin{definition}
    \label{def:Coop_Shap_unif} For a game \(v\) and any \(S \subseteq N \), the {\em Uniform Shapley} coopetition index is defined as
    \begin{multline}
        \label{eq:def_Coop_Shap_unif}
        \mathcal{C}^v_{SU}(S):=\mathcal{C}^v_{p^u,q^r}(S) = \frac{1}{2^{s-1}-1}   \sum_{T \subseteq N \setminus S} \frac{t!(n-s-t)!}{(n-s+1)!} \Bigg(v(S \cup T) - v(T)  \\ -\sum_{\pi=\{S_1,S_2\} \in \Pi_2(S)} \big( v(S_1 \cup T)  + v(S_2 \cup T) - 2v(T) \big) \Bigg).
        \end{multline}  
\end{definition}
Many of the formulas that we have already been examined here, take a simple form. For instance, formula \eqref{eq:def_Coop_Shap_unif} can be written as
\begin{equation}
    \label{eq:newchar_SU}
    \mathcal{C}^v_{SU}(S)=\Phi^v_{Sh}(S)- \frac{1}{2^{s-1}-1}\sum_{\emptyset \neq R \subsetneq S}  \Phi^{-S \setminus R}_{Sh}(R).
\end{equation}

Several properties of the indices just introduced have been considered in the previous work \cite{aleandri2025coopindex} in the context of simple games. These results are valid for TU-games as well, and we list the most relevant among them.

First of all we note that, under both the uniform and the permutation based families, the attitude of a coalition with respect to another coincides when the first coalition is small.

\begin{proposition}
    Take $S\subset N$, with \(|S| \leq 3\). Then, for every $T \subseteq N\setminus S$, the attitude of coalition $S$ towards $T$  based on the uniform and the permutation based families coincide:
    \begin{align*}
        \mathcal{A}^v_{p^u}(S,T) = \mathcal{A}^v_{p^r}(S,T).
    \end{align*}
\end{proposition}
\begin{proof}
    Trivial for \(|S| \leq 2\). For \(|S|=3\) consider the proof of Proposition 4.5 in \cite{aleandri2025coopindex}. 
\end{proof}

However, this proposition does not exclude the introduction of other proposals as it will be clearer later in the work. 
Here we return to a result available in Aleandri and Dall'Aglio \cite{aleandri2025coopindex} in the context of simple games that can be extended to monotone TU-game at no cost and that helps shed some light on the interpretation of the coopetition indices.

\begin{proposition}
For given \(N\), take two families of probability distributions $p^N$ and \(\hat{q} =\big\{q^{-M}:=\{q^{-M}_{R}\}_{R \subseteq N\setminus M}; M \subset N\big\}\), where $q^{-M}_{R}:=q^{N\setminus M}_{R}$, such that for any \(S \subseteq N\) and any 2-partition \(\{S_1,S_2\}\) of \(S\)
\begin{equation}\label{eq:charac} q^N_S(T)=q_{S_1}^{-S_2}(T),\qquad \forall T\subseteq N\setminus S.
\end{equation}
Then
\begin{equation}
    \label{eq:newchar}
    \mathcal{C}^v_{p,q^N}(S)=\Phi^v_{q^N}(S)-\sum_{\pi=\{S_1,S_2\} \in \Pi_2(S)} p_S(\pi)\left( \Phi^{-S_2}_{q^{-S_2}}(S_1) + \Phi^{-S_1}_{q^{-S_1}}(S_2)  \right).
\end{equation}

\end{proposition}
\begin{proof}
    See the proof of Proposition 3.8 in \cite{aleandri2025coopindex}
\end{proof}
The following result can also be adapted to the present context.
\begin{proposition}\label{prop:Dnull}
    Let $i \in N$ be a null player. Take $S\subseteq N\setminus i$ with $|S|\geq 2$ and a pair of families of probability distributions $q$ and $p$  such that $p_S, p_{S\cup  i}, q_S, q_{S\cup  i}$ satisfy
    \begin{enumerate}
        \item  $\dfrac{p_{S\cup  i}(\{S_1\cup i,S_2\}) + p_{S\cup  i}(\{S_1,S_2\cup i\})}{p_{S}(\{S_1,S_2\})}=: K_p(S)$,\quad $\forall\{S_1,S_2\}\in\Pi_2(S)$;
        \item   $\dfrac{q_{S\cup i}(T)}{q_{S}(T) + q_{S}(T\cup i)} =: K_q(S)$,\quad $\forall T\in 2^{N\setminus (S\cup  i)}$.
    \end{enumerate}
    Then
    \begin{equation}
    \label{eq:addnull}
        \mathcal{C}^v_{p,q}(S\cup i) = K_p(S)K_q(S) \mathcal{C}^v_{p,q}(S).
    \end{equation}
\end{proposition}

\begin{proof}
   The proof of Proposition 5.1 in \cite{aleandri2025coopindex} can be used as it is.
\end{proof}

The following result, which is an immediate consequence of the definition and the connection with the Interaction index introduced in \cite{grabisch1999axiomatic}. 
\begin{proposition}\label{prop:DnullS1}
    Let $i \in N$ be a null player and $j\in N\setminus i$, then, for any pair of families of probability distributions $q$ and $p$,
    \begin{equation}
        \mathcal{C}^v_{p,q}(\{j, i\}) = 0.
    \end{equation}   
    
\end{proposition}

\begin{proof}
    Observe that \(\mathcal{A}_p^v(\{j, i\},T) = 0\) for every \(T \subseteq N \setminus \{j, i\}\) and, therefore,
    \(\mathcal{C}_{p,q}^v(\{j, i\})=0
    \).
\end{proof}

This leads to the following corollary.
\begin{corollary}
\label{cor:Dnull}
    Suppose \(q=q^u\) or \(q=q^r\).
    \begin{enumerate}[i)]
        \item If \(p=p^u\), then
        \begin{equation}
            \label{eq:addnull_punif}
            \mathcal{C}^v_{p^u,q}(S\cup i) = \frac{2^{s}-2}{2^{s}-1} \mathcal{C}^v_{p^u,q}(S) = \big( 1 - p^u(\{S,i\}) \big) \mathcal{C}^v_{p^u,q}(S).
        \end{equation}
                \item If \(p=p^r\), then
        \begin{equation}
            \label{eq:addnull_pperm}
            \mathcal{C}^v_{p^r,q}(S\cup i)  = \frac{(s-1)(s+2)}{s(s+1)} \mathcal{C}^v_{p^r,q}(S) = \big( 1 - p^r(\{S,i\}) \big) \mathcal{C}^v_{p^r,q}(S).
        \end{equation}

    \end{enumerate}
\end{corollary}

\section{Axiomatization of the Coopetition Indices}
In this section we propose a set of axioms that are used to characterized the Uniform Shapley Coopetition index and the Shapley-Owen Coopetition index. The two indices share three common axioms and need a fourth one that reflects how to measure the attitude to cooperate and to compete.\\

\begin{definition}
A \emph{group index} is a mapping $\xi$ that assigns to every game $v \in \mathcal{G}^N$ and every $S \subset N$ a real number $\xi(S ; N, v) \in \mathbb{R}$, with $\xi(\varnothing ; N, v)=0$.
\end{definition}

\begin{axiom}[Linearity (L)]\label{ax-L}
A group index $\xi$ satisfies \emph{linearity} if
\begin{equation*}
    \xi(S; N, \alpha v + \beta w) = \alpha \xi(S; N, v) + \beta \xi(S; N, w), 
\end{equation*}
for all $S\subseteq N$, $\alpha,\beta\in\mathbb{R}$, and games $v,w\in\mathcal{G}^N$, where $\alpha v + \beta w\in\mathcal{G}^N$ is given by $(\alpha v + \beta w)(S)=\alpha v(S) + \beta w(S)$ for all $S\subseteq N$.
\end{axiom}

\begin{axiom}[Symmetry over pure bargaining game (SB)]\label{ax-SB}
A group index $\xi$ satisfies \emph{Symmetry over pure bargaining game} if $\xi(S; N, u_N) =\frac{1}{n-s+1}$, for all non-empty $S\subseteq N$, $(N,u_N)$ is the unanimity game with respect to the grand coalition.
\end{axiom}


\begin{axiom}[External Null Player Neutrality (EN)]\label{ax-EN}
A group index $\xi$ satisfies \emph{External Null Player Nullity} if, for any $v\in\mathcal{G}^N$, $S\subset N$ and  $i\notin S$ null player, then
\begin{equation*}
    \xi(S;N,v)=\xi(S;N\setminus i,v^{-i}).
\end{equation*}
\end{axiom}
\begin{axiom}[Internal Null Player Contraction for Uniform-based Attitude (ICU)]\label{ax-ICU}
A group index $\xi$ satisfies \emph{Internal Null Player Contraction} if, for any $v\in\mathcal{G}^N$, $S\subset N$ and  $i\in S$ null player, then
\begin{equation*}
    \xi(S;N,v)=\big(1 - p^u(\{S \setminus i,i\})  \big)\xi(S\setminus i;N\setminus i,v^{-i})=\frac{2^{s-1}-2}{2^{s-1}-1}\xi(S\setminus i;N\setminus i,v^{-i}).
\end{equation*}
\end{axiom}

\begin{axiom}[Internal Null Player Contraction for Permutation-based Attitude (ICP)]\label{ax-ICP}
A group index $\xi$ satisfies \emph{Null Player Contraction for Permutation-based Attitude} if, for any $v\in\mathcal{G}^N$, $S\subset N$ and  $i\in S$ null player, then
\begin{equation*}
    \xi(S;N,v)=\big(1 - p^r(\{S \setminus i,i\})  \big)\xi(S\setminus i;N\setminus i,v^{-i})=\frac{(s+1)(s-2)}{s(s-1)}\xi(S\setminus i;N\setminus i,v^{-i}).
\end{equation*}
\end{axiom}

We start with the characterization of the Uniform Shapley Coopetition index.

\begin{theorem}\label{theo: characterization}
    The Uniform Shapley coopetition index is the unique group indexs that satisfies Axioms \eqref{ax-L} (L), \eqref{ax-SB} (SB), \eqref{ax-EN} (EN), and \eqref{ax-ICU} (ICU).
\end{theorem}

\begin{proof}
First of all, we show that the Shapley Uniform coopetition index satisfies all the four axioms. Axioms 2 (SB) and 4 (EN) are easy and left to the reader. 

Axiom 1 (L) can be easily obtained from formula \eqref{eq:newchar_SU} and from the linearity of each generalized Shapley value. Axiom 3 (ICU) is a consequence of part {\em i)} of Corollary \ref{cor:Dnull}. Now, we show that the four axioms uniquely identify the group index on the unanimity games defined, for each , as
\begin{equation}
    \label{def:unanimity_games}
    u_C(S) = \begin{cases}
        1, & \mbox{if }S \supseteq C;
        \\
        0, & \mbox{otherwise}.
    \end{cases}
\end{equation}
We now compute the uniform Shapley coopetition index for the unanimity game \(u_C\) distinguishing three cases:
\begin{description}
    \item[Case 1] \(S \subseteq N \setminus C\) 

    Here \(\mathcal{C}_{SU}^{u_C}(S)=0\) because all players are null.
    \item[Case 2]\(S \subseteq C\)
    
    Here \(\mathcal{C}_{SU}^{u_C}(S) = \frac{1}{c-s+1}\) because the attitude is 1 only when \(T \supseteq C \setminus S\) and zero otherwise. The probability of the first occurrence under \(q^r\) is \(\frac{1}{c-s+1}\).
    \item[Case 3] \(S \cap C \neq \emptyset\), \(S \cap (N \setminus C) \neq \emptyset\), \(|S \cap C|=r\)
    
    Here \(\mathcal{C}_{SU}^{u_C}(S) = \frac{2^{s-1}-2^{s-r}}{2^{s-1}-1} \frac{1}{c-r+1}\) because the attitude will be 1 only when \(T \supseteq C \setminus S\) and for those 2-partitions of \(S\) in which each subset contains a proper subset of \(S \cap C\). Under \(q^r\), the first event happens with probability \( \frac{1}{c-r+1} \), while the second event, under \(p^u\), happens with probability \(\frac{2^{s-1}-2^{s-r}}{2^{s-1}-1}  \). 
   Indeed, from the sure event we subtract the probability that the partition has null attitude value. These are all the partitions obtained by attaching the whole \(S \cap C\) to each set of the 2-partitions of \(S \cap C^c\), plus the partition \(\{S \cap C, S \cap C^c\}\). The probability is that
        \[
        1 - \frac{(2^{s-r-1}-1)2}{2^{s-1}-1} - \frac{1}{2^{s-1}-1} = \frac{2^{s-1}-2^{s-r}}{2^{s-1}-1}.
        \]
  \end{description}      

We now show that under Axioms \eqref{ax-L} (L), \eqref{ax-SB} (SB), \eqref{ax-ICU} (ICU), and \eqref{ax-EN} (EN), the group index \(\xi\) coincides with the uniform Shapley coopetition index in all three cases.
\begin{description}
    \item[Case 1] A repeated application of Axiom \eqref{ax-ICU} (ICU) yield
    \[
    \xi(S;N,u_C)=\prod_{t=1}^s \frac{2^{s-t}-2}{2^{s-t}-1} \xi(\emptyset;C,u_C^{-N \setminus C}) = 0 = \mathcal{C}_{SU}^{u_C}(S).
    \]
    \item[Case 2] An iterated application of Axiom \eqref{ax-EN} (EN), followed by Axiom \eqref{ax-SB} (SB), yields
    \[
    \xi(S;N,u_C) = \xi(S;C,u_C) = \frac{1}{c-s+1}= \mathcal{C}_{SU}^{u_C}(S).
    \]
\item[Case 3]An iterated application of Axiom \eqref{ax-ICU} (ICU) and Axiom \eqref{ax-EN} (EN) and, followed by Axiom \eqref{ax-SB} (SB), yields
\[
\xi(S;N,u_C) = \prod_{t=1}^{s-r} \frac{2^{s-t}-2}{2^{s-t}-1} \xi(S \cap C,C,u_C^{-N \setminus C}) = 2^{s-r} \frac{2^{c-1}-1}{2^{s-1}-1} \frac{1}{c-r+1} = \mathcal{C}_{SU}^{u_C}(S).
\]
\end{description}
Any TU-game can be written as a linear combination of unanimity games and Axioms \eqref{ax-L} (L) allows to extend the results to all TU-games. 
\end{proof}

We now show that the axioms of Theorem \ref{theo: characterization} are independent. 

\begin{theorem}
    Axiom \ref{ax-L} (L), Axiom \ref{ax-SB} (SB), Axiom \ref{ax-EN} (EN) and Axiom \ref{ax-ICU} (ICU) are logically independent.  
\end{theorem}
\begin{proof}
 {\bf No L:} Take $N=1$ and define the $\xi^{L}(S;N,v)=1$ if $v$ is not a null game or $S\neq\emptyset$ and $\xi^{L}(S;N,v)=0$ otherwise. Take $N\geq2$. For any game $(N,v)$, if there are no null players define $\xi^{L}(S;N,v)=\frac{1}{n-s+1}$, otherwise 
 \begin{equation*}
     \xi^{L}(S;N,v)=\frac{1}{(n-m)-(s-z)+1} \left( \frac{2^{s-1}-2^z}{2^{s-1}-1} \right),
 \end{equation*}
 where $m$ is the number of null players in $(N,v)$ and $z$ is the number of null players in $S$. It is immediate to verify that $\xi^{L}$ does not satisfy Axiom \ref{ax-L} (L). By definition $\xi^{L}$ satisfies Axiom \ref{ax-SB} (SB). Take a game $v$ with $m\geq 0$ null players and a coalition $S\subseteq N$ with $0\leq z\leq m$ null players. Suppose $m=1$ and let $i$ be the null player. If $i\in S$ by definition 
 \begin{equation*}
     \xi^{L}(S;N,v)= \frac{1}{n-s+1} \left( \frac{2^{s-1}-2}{2^{s-1}-1}\right)=\xi^{L}(S\setminus i;N\setminus i,v^{-i})\left( \frac{2^{s-1}-2}{2^{s-1}-1}\right).
 \end{equation*}

 If $i\notin S$ by definition 
 \begin{equation*}
     \xi^{L}(S;N,v)= \frac{1}{n-s+1} \left( \frac{2^{s-1}-1}{2^{s-1}-1}\right)=\xi^{L}(S\setminus i;N\setminus i,v^{-i}).
 \end{equation*}
Suppose $m\geq 2$ and let $i$ be the null player.
If $i\in S$,  by definition
\begin{multline*}
     \xi^{L}(S;N,v)= \frac{1}{(n-m)-(s-z)+1} \left( \frac{2^{s-1}-2^z}{2^{s-1}-1}\right)\\
     =\frac{1}{(n-1-(m-1))-(s-1-(z-1))+1} \left( \frac{2^{s-2}-2^{z-1}}{2^{s-2}-1}\right)\left( \frac{2^{s-1}-2}{2^{s-1}-1}\right)\\
     =\xi^{L}(S\setminus i;N\setminus i,v^{-i})\left( \frac{2^{s-1}-2}{2^{s-1}-1}\right), 
 \end{multline*}
 and $i\notin S$,  by definition 
\begin{multline*}
     \xi^{L}(S;N,v)= \frac{1}{(n-m)-(s-z)+1} \left( \frac{2^{s-1}-2^z}{2^{s-1}-1}\right)\\
     =\frac{1}{(n-1-(m-1))-(s-z))+1} \left( \frac{2^{s-1}-2^{z}}{2^{s-1}-1}\right)    =\xi^{L}(S\setminus i;N\setminus i,v^{-i}).
 \end{multline*}
Then Axiom  \eqref{ax-ICU} (ICU)  and  Axiom \eqref{ax-EN} (EN) are satisfied.

\medskip 
 {\bf No SB:}
For any $(N,v)$ and $S\subseteq N$ define $\xi^{SB}(S;N,v)=0$. It is immediate to verify that $\xi^{SB}$ satisfies all the axioms, but not Axiom \ref{ax-SB} (SB). 
\medskip
\\
 {\bf No INPCU:}
For any $(N,v)$ and $S\subseteq N$ define $\xi^{INPCU}(S;N,v)=\Phi_{Sh}(S;N,v)$ the generalized Shapley value. It is immediate to verify that $\xi^{INPCU}$ satisfies all the axioms, but not Axiom \ref{ax-ICU} (ICU). 
\medskip
\\
 {\bf No ENPN:} For any $\emptyset \neq C\subsetneq N$, define 

 \begin{equation*}
    \xi^{ENPN}(S;N,u_C)=\begin{cases}
        \frac{1}{n-s+1} \left( \frac{2^{s-1}-2^{s-r}}{2^{s-1}-1} \right), & \text{if } S\cap C\neq \emptyset \text{ and } r=|S\cap C| < s;\\
        0, & \text{otherwise}. 
    \end{cases} 
 \end{equation*}
Moreover, define $\xi^{ENPN}(S;N,u_N)=\frac{1}{n-s+1}$ and extend the solution by linearity. By construction $\xi^{NPN}$ satisfies Axioms \eqref{ax-L} (L) and \eqref{ax-SB} (SB). Arguing as in the first step of the proof we obtain that $\xi^{ENPN}$ satisfies Axiom \eqref{ax-ICU} (ICU), but not  Axiom \eqref{ax-EN} (EN). Indeed, if a null player $i$ does not belong to the coalition $S$, by definition 
$$
\xi^{ENPN}(S\setminus i; N\setminus i,u_C^{-i})=\frac{1}{n-1-s+1} \left( \frac{2^{s-1}-2^{s-r}}{2^{s-1}-1} \right) \neq \xi^{ENPN}(S;N,u_C).
$$
\end{proof}

We now consider the Shapley-Owen coopetition index obtained by considering the family of distributions \(p^r\). We axiomatize this index by changing Axiom (ICU) \ref{ax-ICU} with Axiom (ICP) \ref{ax-ICP}.

\begin{theorem}\label{theo: charShapOwen}
    The Shapley-Owen coopetition index is the unique group indexs that satisfies Axioms \eqref{ax-L} (L), \eqref{ax-SB} (SB), \eqref{ax-EN} (EN), and \eqref{ax-ICP} (ICP).
\end{theorem}

\begin{proof}
The proof retraces the steps of the proof of Theorem \ref{theo: characterization}, but for completeness we write it explicitly. First of all we show that the Shapley Uniform coopetition index satisfies all the four axioms. Axioms \eqref{ax-SB} (SB) and \eqref{ax-EN} (EN) are easy and left to the reader (but we may write them if we have time). Axiom \eqref{ax-L} (L) can be easily obtained from formula \eqref{eq:newchar_SU} and from the linearity of each group Shapley value. Axiom \eqref{ax-ICP} (ICP) is a consequence of part {\em ii)} from Corollary \ref{cor:Dnull}. In order to prove that the four axioms uniquely identify the group index on the unanimity games  we need the following technical result. 
\begin{lemma}
    For any positive integers \(r < s\)
    \begin{equation}
        \label{lem:combinatorial}
            1 - \frac{2r!(s-r)!}{s!(s-1)} - \sum_{t=1}^{s-r-1} \binom{s-r}{t} \frac{t!(s-t)!+(t+r)!(s-t-r)!}{s!(s-1)} = \frac{(s+1)(r-1)}{(s-1)(r+1)}.
    \end{equation}
\end{lemma}
\begin{proof} 
Take two positive integers $r<s$. First, observe that, for each \(t\leq s-r-1\),
\begin{equation}\label{eq:riso0}
\sum_{t=1}^{s-r-1}\binom{s-r}{t}\frac{(t+r)!(s-t-r)!}{s!(s-1)}
   =\frac{(s-r)!}{s!(s-1)}\sum_{t=1}^{s-r-1}\prod_{\ell=1}^r(t+\ell),
\end{equation}
and
\begin{equation}\label{eq:riso1}
    \sum_{t=1}^{s-r-1}\binom{s-r}{t}\frac{(s-t)!(t)!}{s!(s-1)}
   =\frac{(s-r)!}{s!(s-1)}\sum_{t=1}^{s-r-1}\prod_{\ell=0}^{r-1}(s-t-\ell).
\end{equation}
Take $t'=s-r-t$ and $\ell'=r-\ell$, then
\begin{equation}\label{eq:riso11}
    \sum_{t=1}^{s-r-1}\prod_{\ell=0}^{r-1}(s-t-\ell) = \sum_{t'=s-r-1}^{1}\prod_{\ell=0}^{r-1}(t'+r-\ell) = \sum_{t'=1}^{s-r-1}\prod_{\ell'=1}^{r}(t'+\ell').
\end{equation}
Define
\[
E(s,r):= 1
-\frac{2\,r!\,(s-r)!}{s!\,(s-1)}
-\sum_{t=1}^{s-r-1}
\binom{s-r}{t}
\frac{t!(s-t)!+(t+r)!(s-t-r)!}{(s-1)\,s!}\,.
\]
Using \eqref{eq:riso0},\eqref{eq:riso1}, and \eqref{eq:riso11} we obtain
\[
E(s,r)=1-\frac{2(s-r)!}{s!(s-1)}\Big(r!+\sum_{t=1}^{s-r-1}\prod_{\ell=1}^r(t+\ell)\Big) = 1-\frac{2(s-r)!}{s!(s-1)}\sum_{t=0}^{s-r-1}\prod_{\ell=1}^r(t+\ell).
\]
Moreover
\[
\sum_{t=0}^{s-r-1}\prod_{\ell=1}^r(t+\ell) = \sum_{t=0}^{s-r-1}\frac{(t+r)!}{t!}=r!\sum_{t=0}^{s-r-1}\binom{t+r}{t}=r!\binom{s}{s-r-1},
\]
where the last equality follows from the combinatorial formula 5.9 in \cite{Graham1994concrete}.
Then
\begin{equation*}
    E(s,r)=1-\frac{2(s-r)}{(s-1)(r+1)}=\frac{(s+1)(r-1)}{(s-1)(r+1)}.
\end{equation*}
\end{proof}

We now compute the Shapley-Owen coopetition index for the unanimity game \(u_C\) distinguishing three cases:
\begin{description}
    \item[Case 1] \(S \subseteq N \setminus C\) 

    Here \(\mathcal{C}_{SO}^{u_C}(S)=0\) because all players are null.
    \item[Case 2]\(S \subseteq C\)
    
    Here \(\mathcal{C}_{SO}^{u_C}(S) = \frac{1}{c-s+1}\) because the attitude is 1 only when \(T \supseteq C \setminus S\) and zero otherwise. The probability of the first occurrence under \(q^r\) is \(\frac{1}{c-s+1}\).
    \item[Case 3] \(S \cap C \neq \emptyset\), \(S \cap (N \setminus C) \neq \emptyset\), \(|S \cap C|=r\)
    
    Here \(\mathcal{C}_{SO}^{u_C}(S) = \frac{(s+1)(r-1)}{(s-1)(r+1)(c-r+1)} \) because the attitude will be 1 only for when \(T \supseteq C \setminus S\) and for those 2-partitions of \(S\) in which each subset contains a proper subset of \(S \cap C\). Under \(q^r\), the first event happens with probability \( \frac{1}{c-r+1} \), while the second event, under \(p^u\), happens with probability \(\frac{(s+1)(r-1)}{(s-1)(r+1)}  \). We use the same method  from Case 3 of Theorem \ref{theo: characterization}, to have
    \[
    1 - \frac{2r!(s-r)!}{s!(s-1)} - \sum_{t=1}^{s-r-1} \binom{s-r}{t} \frac{t!(s-t)!+(t+r)!(s-t-r)!}{s!(s-1)}, 
    \]
    that, by Lemma \ref{lem:combinatorial}, equals \(\frac{(s+1)(r-1)}{(s-1)(r+1)}\).
\end{description}
We now show that under Axioms \eqref{ax-L} (L), \eqref{ax-SB} (SB), \eqref{ax-ICP} (ICP), and \eqref{ax-EN} (EN), the group index \(\xi\) coincides with the Shapley-Owen coopetition index in all three cases.
\begin{description}
    \item[Case 1] A repeated application of \eqref{ax-ICP} (ICP) yield
    \[
    \xi(S;N,u_C)=\prod_{t=s-r}^s \frac{(t+1)(t-2)}{t(t-1)} \xi(\emptyset,C,u_C^{-N \setminus C}) = 0 = \mathcal{C}_{SO}^{u_C}(S).
    \]
    \item[Case 2] An iterated application of \eqref{ax-EN} (EN), followed by \eqref{ax-SB} (SB), yields
    \[
    \xi(S;N,u_C) = \xi(S;C,u_C) = \frac{1}{c-s+1}= \mathcal{C}_{SO}^{u_C}(S).
    \]
\item[Case 3]An iterated application of Axioms \eqref{ax-ICP} (ICP), and \eqref{ax-EN} (EN) and, followed by \eqref{ax-SB} (SB), yields
\[
\xi(S;N,u_C) = \prod_{t=s-r}^s \frac{(t+1)(t-2)}{t(t-1)} \xi(S \cap C;C,u_C^{-N \setminus C}) = \frac{(s+1)(r-1)}{(s-1)(r+1)} \frac{1}{c-r+1} = \mathcal{C}_{SO}^{u_C}(S).
\]
\end{description}
Any TU-game can be written as a linear combination of unanimity games and Axiom \eqref{ax-L} (L) allows to extend the results to all TU-games. 
\end{proof}

Finally, we show that, just as in the uniform case, the axioms are independent. The proof is quite similar to the previous one -- with the due changes for the contraction factors. In any case, we give the full proof here.

\begin{theorem} \label{th:indep_permbased}
    Axiom \ref{ax-L} (L), Axiom \ref{ax-SB} (SB), Axiom \ref{ax-EN} (EN) and Axiom \ref{ax-ICP} (ICP) are logically independent  
\end{theorem}
\begin{proof}
 {\bf No L:} Take $N=1$ and define the $\xi^{L}(S;N,v)=1$ if $v$ is not a null game or $S\neq\emptyset$ and $\xi^{L}(S;N,v)=0$ otherwise. Take $N\geq2$. For any game $(N,v)$, if there are no null players define $\xi^{L}(S;N,v)=\frac{1}{n-s+1}$, otherwise 
 \begin{equation*}
     \xi^{L}(S;N,v)=\frac{1}{(n-m)-(s-z)+1} \frac{(s+1)(s-z-1)}{(s-1)(s-z+1)} ,
 \end{equation*}
 where $m$ is the number of null players in $(N,v)$ and $z$ is the number of null players in $S$. It is immediate to verify that $\xi^{L}$ does not satisfy Axiom \ref{ax-L} (L). By definition $\xi^{L}$ satisfies Axiom \ref{ax-SB} (SB). Take a game $v$ with $m\geq 0$ null players and a coalition $S\subseteq N$ with $0\leq z\leq m$ null players. Suppose $m=1$ and let $i$ be a null player. If $i\in S$, by definition 
 \begin{equation*}
     \xi^{L}(S;N,v)= \frac{1}{n-s+1} \frac{(s+1)(s-2)}{(s-1)s}=\xi^{L}(S\setminus i;N\setminus i,v^{-i})\frac{(s+1)(s-2)}{(s-1)s}.
 \end{equation*}
 If $i\notin S$, by definition 
 \begin{equation*}
     \xi^{L}(S;N,v)= \frac{1}{n-s+1} \frac{(s+1)(s-1)}{(s-1)(s+1)}=\xi^{L}(S\setminus i;N\setminus i,v^{-i}).
 \end{equation*}
Suppose $m\geq 2$ and let $i$ be a null player.
If $i\in S$,  by definition
\begin{multline*}
     \xi^{L}(S;N,v)= \frac{1}{(n-m)-(s-z)+1} \frac{(s+1)(s-z-1)}{(s-1)(s-z+1)}\\
     =\frac{1}{(n-1-(m-1))-(s-1-(z-1))+1} \frac{(s)(s-1-(z-1)-1)}{(s-1)(s-z+1)}\frac{(s-2)(s+1)}{s(s-1)}\\
     =\xi^{L}(S\setminus i;N\setminus i,v^{-i})\frac{(s-2)(s+1)}{s(s-1)}, 
 \end{multline*}
 and $i\notin S$,  by definition 
\begin{multline*}
     \xi^{L}(S;N,v)= \frac{1}{(n-m)-(s-z)+1} \frac{(s+1)(s-1)}{(s-1)(s+1)}\\
     =\frac{1}{(n-1-(m-1))-(s-z))+1} \frac{(s+1)(s-1)}{(s-1)(s+1)} =\xi^{L}(S\setminus i;N\setminus i,v^{-i}).
 \end{multline*}
Then Axiom  \eqref{ax-ICP} (ICU)  and  Axiom \eqref{ax-EN} (EN) are satisfied.
\medskip 
\\
 {\bf No SB:}
For any $(N,v)$ and $S\subseteq N$ define $\xi^{SB}(S;N,v)=0$. It is immediate to verify that $\xi^{SB}$ satisfies all the axioms, but Axiom \ref{ax-SB} (SB). 
\medskip
\\
 {\bf No INPCP:}
For any $(N,v)$ and $S\subseteq N$ define $\xi^{INPCP}(S;N,v)=\Phi_{Sh}(S;N,v)$ the generalized Shapley value. It is immediate to verify that $\xi^{INPC}$ satisfies all the axioms, but not Axiom \ref{ax-ICP} (ICP). 
\medskip
\\
 {\bf No ENPN:} For any $\emptyset \neq C\subsetneq N$, define 

 \begin{equation*}
    \xi^{ENPN}(S;N,u_C)=\begin{cases}
        \frac{1}{(n-s+1)} \frac{(s+1)(r-1)}{(s-1)(r+1)}, & \text{if } S\cap C\neq \emptyset \text{ and } r=|S\cap C| < s;\\
        0, & \text{otherwise}. 
    \end{cases} 
 \end{equation*}
Moreover, define $\xi^{ENPN}(S,N,u_N)=\frac{1}{n-s+1}$ and extend the solution by linearity. By construction $\xi^{NPN}$ satisfies Axioms \eqref{ax-L} (L) and \eqref{ax-SB} (SB). Arguing as in the first step of the proof we obtain that $\xi^{ENPN}$ satisfies Axiom \eqref{ax-ICP} (ICP), but not  Axiom \eqref{ax-EN} (EN). Indeed, if a null player $i$ does not belong to the coalition $S$, by definition 
$$
\xi^{ENPN}(S\setminus i;N\setminus i,u_C^{-i})=\frac{1}{(n-1-s+1)} \frac{(s+1)(r-1)}{(s-1)(r+1)} \neq \xi^{ENPN}(S;N,u_C).
$$
\end{proof}

\section{Conclusions}
This work generalizes the coopetition index to monotone TU-games and to coalitions of all sizes, including singletons. The extension clarifies the relationship between coopetition and classical group values, showing that the index is bounded by the corresponding generalized semivalue and coincides with it for individual players. The analysis identifies structural conditions under which a coalition behaves as fully cooperative or fully antagonistic, and introduces an absolute, scale-free measure of coopetition.

This work also provides an axiomatic characterization of two notable cases: the Uniform Shapley coopetition index and the Shapley–Owen coopetition index. Axioms for Linearity, Symmetry over pure bargaining, and External Null Player Nullity are used in both axiomatizations. The former is standard in the literature, while the second reflects the fact that when all players are essential to obtain a positive value, the members of a coalition 
\(S\) are forced to cooperate and act as a single player; their power is then divided uniformly among the “active” players. The third axiom requires that the power of a coalition is not affected when a null player external to the coalition is removed from the game. A null-player contraction axiom, specific to each index, characterizes the type of solution by specifying that, when a null player of a coalition is removed from the game, the power of the coalition scales proportionally to the probability that a null player is isolated within the coalition.
 

\bibliographystyle{abbrv}
\bibliography{references}

\end{document}